\documentclass[journal]{IEEEtran}

\usepackage{amsfonts}
\usepackage{color}
\usepackage{graphicx}
\usepackage[dvips]{epsfig}
\usepackage{graphics} 
\usepackage{times} 
\usepackage[cmex10]{amsmath} 
\usepackage{amssymb}  
\usepackage{cite}
\usepackage{multirow}
\usepackage[tight,footnotesize]{subfigure}
\usepackage{amsmath}
\usepackage[boxed,ruled,lined]{algorithm2e}

\def\norm #1{\left\|#1\right\|}
\def\inftyn #1{\left\|#1\right\|_{\infty}}
\def\twon #1{\left\|#1\right\|_2}
\def\onen #1{\left\|#1\right\|_1}
\def\zeron #1{\left\|#1\right\|_0}

\def\sgn #1{\text{sgn}#1}
\def\abs #1{\left|#1\right|}
\def\inp #1{\left\langle#1\right\rangle}

\def\st{\textrm{ subject to }}

\def\bR{\mathbb{R}}

\def\m #1{\boldsymbol{#1}}

\def\cD{\mathcal{D}}

\def\cI{\mathcal{I}}

\def\cK{\mathcal{K}}

\def\cN{\mathcal{N}}
\def\cP{\mathcal{P}}
\def\cQ{\mathcal{Q}}

\def\bee{\begin{equation}}
\def\ene{\end{equation}}

\def\beq{\begin{eqnarray}}
\def\enq{\end{eqnarray}}
\def\lentwo{\setlength\arraycolsep{2pt}}

\newtheorem{lem}{Lemma}
\newtheorem{rem}{Remark}

\def\equ #1{\begin{equation}#1\end{equation}}
\def\equa #1{\begin{eqnarray}#1\end{eqnarray}}
\def\sbra #1{\left(#1\right)}
\def\mbra #1{\left[#1\right]}
\def\lbra #1{\left\{#1\right\}}
\def\diag #1{\text{diag}#1}

\def\st {\text{ subject to }}

\begin{document}

\title{Variational Bayesian Algorithm for Quantized Compressed Sensing}
\author{Zai Yang, Lihua Xie$^*$, {\em Fellow, IEEE,} and Cishen Zhang

\thanks{$^*$Author for correspondence.

Z. Yang and L. Xie are with EXQUISITUS, Centre for E-City, School of Electrical and Electronic Engineering, Nanyang Technological University, 639798, Singapore (e-mail: yang0248@e.ntu.edu.sg; elhxie@ntu.edu.sg).

C. Zhang is with the Faculty of Engineering and Industrial Sciences, Swinburne University of Technology, Hawthorn VIC 3122, Australia (e-mail: cishenzhang@swin.edu.au).}}


%


\maketitle

\begin{abstract} Compressed sensing (CS) is on recovery of high dimensional signals from their low dimensional linear measurements under a sparsity prior and digital quantization of the measurement data is inevitable in practical implementation of CS algorithms. In the existing literature, the quantization error is modeled typically as additive noise and the multi-bit and 1-bit quantized CS problems are dealt with separately using different treatments and procedures. In this paper, a novel variational Bayesian inference based CS algorithm is presented, which unifies the multi- and 1-bit CS processing and is applicable to various cases of noiseless/noisy environment and unsaturated/saturated quantizer. By decoupling the quantization error from the measurement noise, the quantization error is modeled as a random variable and estimated jointly with the signal being recovered. Such a novel characterization of the quantization error results in superior performance of the algorithm which is demonstrated by extensive simulations in comparison with state-of-the-art methods for both multi-bit and 1-bit CS problems.

\end{abstract}




%

\section{Introduction}
The recently developed compressed sensing (CS) theory and methods \cite{candes2006compressive,donoho2006compressed} can achieve acquisition of information contained within a huge volume of data using only a small number of measurement samples. Different from the classical Shannon-Nyquist sampling theorem which requires that the sampling frequency be twice as high as the bandwidth of a signal in order to reconstruct its complete information, the CS theory accesses the success of signal recovery with the sparsity. A signal $\m{x}\in\bR^N$ of length $N$ is called $K$-sparse in a basis $\m{\Psi}\in\bR^{N\times N}$ if all but at most a number of $K\ll N$ entries of its coefficient vector $\m{\theta}\in\bR^N$ are zero with $\m{x}=\m{\Psi}\m{\theta}$. Without loss of generality we assume that $\m{\Psi}$ is an identity matrix, i.e., $\m{x}$ is sparse in the canonical basis, since for a general basis $\m{\Psi}$ it can be absorbed into the following introduced sensing matrix $\m{A}$. Rather than observing directly the original sparse signal $\m{x}$, a number of $M$, $K<M\ll N$, linear measurements are acquired in CS as
\equ{\m{y}=\m{A}\m{x}+\m{n}, \label{formu:observation}}
where $\m{y}\in\bR^M$ is the measurement vector, $\m{A}\in\bR^{M\times N}$ denotes the sensing/measurement matrix and $\m{n}\in\bR^M$ is the measurement noise vector. Though the recovery of $\m{x}$ from $\m{y}$ is generally ill-posed (less linear equations than the unknown variables), it is shown in \cite{candes2006stable} that a sparse signal $\m{x}$ can be stably recovered under mild conditions on $\m{A}$ in the sense that the recovery error grows linearly with the noise level. To do this, the basis pursuit denoising (BPDN) problem
\equ{\min\onen{\tilde{\m{x}}},\st\twon{\m{y}-\m{A}\tilde{\m{x}}}\leq\epsilon \label{formu:BPDN}}
is solved where $\epsilon\geq\twon{\m{n}}$ indicates the noise level. The recovery is exact in the noise free case. A similar result holds for compressible signals that are not exactly sparse.

Sparse Bayesian learning (SBL) \cite{tipping2001sparse,wipf2004sparse,ji2008bayesian} was derived from the research area of machine learning and has become a popular method for sparse signal recovery in CS. In SBL, the sparse signal recovery problem is formulated from a Bayesian perspective while the sparsity information is exploited by assuming a sparse prior for the signal of interest. As an example, a Laplace prior \cite{tibshirani1996regression,babacan2010bayesian} corresponds to the $\ell_1$ norm which has been widely studied in existing optimization approaches. Since exact Bayesian inference is typically intractable, approximation approaches to Bayesian inference have been adopted including evidence procedure \cite{mackay1992bayesian}, e.g., in \cite{babacan2010bayesian}, and variance message passing (VMP) \cite{winn2006variational}, e.g., in \cite{pedersen2011sparse}. One merit of Bayesian CS is the flexibility of modeling sparse signals that can not only promote the sparsity of its solution, but also exploit additionally known structures of the sparse signal, see, e.g., \cite{yang2012bayesian,he2009exploiting}. Since the Bayesian inference is a probabilistic method and based on heuristics to some extent, one shortcoming of Bayesian approaches is that there have been fewer results on their signal recovery accuracy in comparison with deterministic approaches, e.g., BPDN.

The conventional CS framework is mainly focused on the sparse signal recovery from the real-valued measurement $\m{y}$ that has infinite bit precision. The required number of measurements $M$ is mainly studied for guaranteed signal recovery accuracy \cite{donoho2011noise,stojnic2009various,yang2012phase,maleki2013asymptotic}. Since quantization is necessary for practical considerations, e.g., data storage and transmission, we study the sparse signal recovery from quantized measurements in this paper. During the quantization process, each continuous-valued measurement is quantized into some value in a finite set. A new challenge is thus the existence of quantization errors. The noise free case with a uniform unsaturated quantizer is studied in \cite{candes2006near,candes2006stable,jacques2011dequantizing}. A solver with quantization consistency is recommended in \cite{candes2006near} that corresponds to replacing the $\ell_2$ norm in BPDN by the $\ell_\infty$ norm. A BPDN solver is used in \cite{candes2006stable} which treats the quantization errors as additive noises with bounded energy. A family of solvers, named as basis pursuit dequantizer of moment $p$ (BPDQ$_p$), that includes BPDN and that in \cite{candes2006near} as special cases is studied in \cite{jacques2011dequantizing} where the $\ell_2$ norm in BPDN is replaced by an $\ell_p$ norm with $2\leq p\leq+\infty$. By characterizing the quantization errors as independent random variables uniformly distributed in a common interval, it is shown in \cite{jacques2011dequantizing} that the optimal signal recovery accuracy is obtained at some finite $p\geq2$. But unfortunately, the optimal $p$ cannot be explicitly given in practice.
Note that the common uniform distribution assumption is crucial to obtain the results in \cite{jacques2011dequantizing}. As a result, it is unclear whether the results in \cite{jacques2011dequantizing} can be extended to a general quantizer case where such an assumption fails. It is obvious that both BPDN and BPDQ$_p$ are inappropriate in the case of a saturated quantizer since data saturation may lead to large or even unbounded quantization errors that deteriorate their performance. To deal with the data saturation, Laska {\em et al.} \cite{laska2011democracy} propose two modified versions of BPDN which either reject saturated measurements or incorporate them into signal recovery. While quantization errors and measurement noises are coupled in most existing methods (some methods, e.g., BPDQ$_p$, consider only the noise free case to avoid such a problem), e.g., in \cite{laska2011democracy}, they are separately studied by Zymnis {\em et al.} \cite{zymnis2010compressed} where the authors seek to find a signal estimate that maximizes the likelihood of the quantized measurements while the $\ell_1$ norm is used to promote the signal sparsity. The resulting algorithm is quoted as $\ell_1$-regularized maximum likelihood (L1RML).

An extreme case of quantized CS is so-called 1-bit CS where each quantized measurement keeps only the sign information of the real-valued measurement and thus uses just one bit. The 1-bit CS framework is proposed in \cite{boufounos2008bit} and has attracted many research interests because it possesses many merits. For example, a 1-bit quantizer is a simple comparator that tests whether the measurement is above or below zero, leading to an easy implementation and a fast quantization process. A measurement noise can be neglected in 1-bit CS as long as it does not change the sign of the measurement. It is shown in \cite{laska2012regime} that to acquire just one bit for each measurement is optimal in the presence of heavy noises. The 1-bit case is quite different from the multi-bit case since all measurements are saturated in 1-bit CS and the signal scaling information is lost. A common approach to the signal scaling problem is to impose that the signal to be recovered has a fixed unit norm and then search for the signal on the unit hyper-sphere rather than in the whole space. Such a constraint is nonconvex and brings new challenges to algorithm design. Existing algorithms based on this constraint include renormalized fixed point iteration (RFPI) \cite{boufounos2008bit}, matching sign pursuit (MSP) \cite{boufounos2009greedy}, restricted-step shrinkage (RSS) \cite{laska2011trust} and binary iterative hard thresholding (BIHT) \cite{jacques2011robust}. Convex formulations of the 1-bit CS problem have been recently proposed by Plan and Vershynin \cite{plan2011one,plan2012robust}. They show in \cite{plan2011one} that a linear program can decode the noiseless case with guaranteed signal recovery accuracy under similar mild conditions as in conventional CS. In \cite{plan2012robust} they introduce a seemingly unrelated convex program for the noisy case and show similar results. It is noted that both the BIHT and the convex program in \cite{plan2012robust} that deal with the noisy 1-bit CS problem require the signal sparsity information (BIHT needs the signal sparsity $K$ and CVXP requires a proper upper bound for the signal's $\ell_1$ norm).

In this paper, we introduce a Bayesian framework for quantized CS that unifies the multi- and 1-bit cases. The new framework deals with quantization errors and measurement noises separately, allows data saturation in multi-bit CS, and does not need the signal sparsity information. Based on the new problem formulation, we propose an algorithm within the Bayesian CS framework where the quantization errors are modeled as random variables and jointly estimated with the signal of interest. A three-layer hierarchical prior introduced in \cite{pedersen2011sparse} is adopted as the sparse signal prior and variational Bayesian inference is carried out using VMP. The performance of the proposed algorithm is studied by extensive numerical simulations in various scenarios. It is shown that the new algorithm improves the signal recovery accuracy in comparison with state-of-the-art methods in both multi- and 1-bit CS. Part of the results of this paper have been presented in \cite{yang2012accurate}.

Notations used in this paper are as follows. Bold-face letters are reserved for vectors and matrices. For ease of exposition, we do not distinguish a random variable from its numerical value. $x_i$ is the $i$th entry of a vector $\m{x}$. $\m{x}_{\cI}$ denotes a truncated vector of $\m{x}$ with entry indices in a set $\cI$. $\m{A}_i$ is the $i$th column of a matrix $\m{A}$. $\zeron{\m{x}}$ counts the number of nonzero entries of a vector $\m{x}$. $\norm{\m{x}}_p=\sbra{\sum_i \abs{x_i}^p}^{1/p}$ denotes the $\ell_p$ norm of a vector $\m{x}$ with $1\leq p\leq+\infty$. $\inp{g\sbra{x}}_{p\sbra{x}}$ denotes the expectation of a function $g\sbra{x}$ with respect to a random variable $x$ whose probability density function is $p\sbra{x}$. $\odot$ denotes the Hardamard (elementwise) product. $\succcurlyeq$ and $\preccurlyeq$ denote $\geq$ and $\leq$ respectively with an elementwise operation.

The rest of the paper is organized as follows. Section \ref{sec:model} introduces the new framework for quantized CS and studies its relations with existing formulations. Section \ref{sec:QVMP} introduces the proposed Q-VMP algorithm. Section \ref{sec:simulation} presents numerical simulations to illustrate the improved signal recovery accuracy of the proposed Q-VMP algorithm in comparison with existing ones. Section \ref{sec:conclusion} concludes the paper and discusses some future works.

\section{A New Framework for Quantized CS} \label{sec:model}
Multi- and 1-bit CS problems are typically studied separately in the literature due to their big difference. In this section, we propose a Bayesian framework that unifies both cases. The new framework is applicable to various scenarios including noiseless/noisy environment and unsaturated/saturated quantizer. Its relations with existing methods are studied through a maximum {\em a posteriori} (MAP) interpretation.

\subsection{A Unified Observation Model}
In quantized CS, the observed samples are noisy linear measurements of the original signal after quantization:
\equ{\m{z}=\cQ\sbra{\m{y}},\quad\m{y}=\m{A}\m{x}+\m{n} \label{formu:OriginObservModel}}
where $\m{x}$ is the signal of interest, $\m{A}$ is the sensing matrix, $\m{n}$ is the measurement noise vector, $\m{y}$ is the pre-quantized noisy measurement vector, $\cQ$ denotes a quantizer and $\m{z}$ is the observation. A quantizer $\cQ\sbra{v}$ for a scalar $v\in\bR$ is defined as
\equ{\cQ\sbra{v}=\left\{\begin{array}{ll}v_0,& \text{if }v\in\sbra{u_0,u_1},\\ v_1,& \text{if }v\in\left[u_1,u_2\right),\\\cdots,&\cdots,\\v_{L-1}, &\text{if }v\in\left[u_{L-1},u_L\right), \end{array}\right. \label{formu:quantizer}}
where $L$ denotes the number of the quantization levels and typically satisfies $L=2^B$ with $B$ denoting the bit depth (bits per quantized measurement), $u_0<u_1<\cdots<u_{L}$, and $v_i\in\left[u_{i},u_{i+1}\right)$ for $i=0,\cdots,L-1$. The quantizer $\cQ\sbra{v}$ is called unsaturated if $\sbra{u_0,u_{L}}$ is a finite interval, or saturated otherwise. For a vector $\m{v}$, $\cQ\sbra{\m{v}}$ operates elementwise. Multi-bit CS refers to the case $B\geq2$ while 1-bit CS corresponds to $B=1$.

\subsubsection{Multi-bit CS}
We consider first a multi-bit quantizer where $B\geq2$. Denote $\cD_{y}$ the domain of $\m{y}$. Then we have
\equ{\cD_{y}=\cQ^{-1}\sbra{\m{z}}:=\lbra{\m{y}\in\bR^M|\cQ\sbra{\m{y}}=\m{z}}.\label{formu:domain_y}}
We introduce an auxiliary variable $\m{e}=\m{z} - \m{y}$ denoting the quantization error with its domain
\equ{\cD_{e}=\m{z}-\cD_y:=\lbra{\m{z}-\m{y}|\m{y}\in\cD_y}. \label{formu:domain_e}}
Note that $\cD_{e}$ is unbounded when data saturation occurs.

\subsubsection{1-bit CS}
In the case of 1-bit quantizer we set $u_0=-\infty$, $u_1=0$ and $u_2=+\infty$. The sign information of $\m{y}$ is preserved in the quantized measurement $\m{z}$. But the scaling information of $\m{y}$ and that of $\m{x}$ is lost. Without loss of generality, we let the 1-bit quantizer
\[\cQ\sbra{v}=\varsigma\sgn\sbra{v}\]
for a scalar $v\in\bR$ with $\varsigma\rightarrow0_+$ ($\varsigma$ is an arbitrarily small positive number) and $\sgn\sbra{\cdot}$ being the sign function. For convenience, we set $\sgn\sbra{0}=1$ (the choice is arbitrary and can be replaced by $\sgn\sbra{0}=-1$). Then we have $\m{z}\rightarrow\m{0}$. To solve the signal scaling problem we impose a constraint that $\m{y}$ has fixed unit norm, i.e.,
\equ{\norm{\m{y}}_s=1 \label{formu:norm_y_1bit}}
with $s\geq1$. Different from the multi-bit quantizer case we have in such a case that
\lentwo\equa{\cD_{y}
&=&\lbra{\m{y}\in\bR^M|\sgn\sbra{\m{y}}=\sgn\sbra{\m{z}},\norm{\m{y}}_s=1},\label{formu:domain_y_1bit}\\ \cD_e
&=&\lbra{\m{e}\in\bR^M|\sgn\sbra{\m{e}}=-\sgn\sbra{\m{z}},\norm{\m{e}}_s=1}. \label{formu:domain_e_1bit}}

As a result, an observation model that unifies the multi- and 1-bit CS problems can be written into
\equ{\m{z}=\m{A}\m{x}+\m{e}+\m{n},\quad \m{e}\in\cD_e, \label{formu:ObservationModel}}
which is the observation model to be used in this paper to recover $\m{x}$.

\begin{rem} In practice, one is able to know the domain of $\m{e}$ but it is difficult to characterize its exact relationship with $\m{x}$. Thus, the dependence of $\m{e}$ on $\m{x}$ is dropped when we write (\ref{formu:OriginObservModel}) into (\ref{formu:ObservationModel}) which is the observation model we use for the signal recovery, i.e., the only information we attempt to exploit during the signal recovery process is its domain. We note that the signal recovery performance may be further improved if the dependence can be properly exploited which, however, is rather difficult.
 \label{rem:modelredundancy}
\end{rem}

\subsection{Bayesian Formulation of Quantized CS}\label{sec:sparseBayesianmodel}
In this subsection we formulate the quantized CS problem from a Bayesian perspective based on the observation model in (\ref{formu:ObservationModel}). According to Remark \ref{rem:modelredundancy} we treat $\m{e}$ as a random variable independent of $\m{x}$. The joint probability density function (PDF) $p\sbra{\m{z},\m{x},\m{e}}$ is decomposed as
\[p\sbra{\m{z},\m{x},\m{e}}=p\sbra{\m{z}|\m{x},\m{e}}p\sbra{\m{x}}p\sbra{\m{e}}.\]
We define the three distributions on the right hand side as follows.

\subsubsection{Noise model} Under an assumption of white Gaussian measurement noise, i.e., $\m{n}\sim\cN\sbra{\m{0},\sigma^2\m{I}}$ where $\sigma^2$ is the noise variance and $\m{I}$ denotes an identity matrix of proper dimension, we have
\equ{p\sbra{\m{z}|\m{x},\m{e};\sigma^2}=\cN\sbra{\m{z}|\m{A}\m{x}+\m{e},\sigma^2\m{I}}. \label{formu:prior_n}}

\subsubsection{Sparse signal model} A sparse prior is needed for the sparse signal $\m{x}$ of interest. Here we do not give an explicit distribution to the sparse signal $\m{x}$ but denote $p\sbra{\m{x}}$ its PDF. Then we let $f\sbra{\m{x}}=-C_1\log p\sbra{\m{x}}+C_2$ where $C_1$ and $C_2$ are proper constants. The only thing that we assume for $p\sbra{\m{x}}$ is that it favors entries of $\m{x}$ being zeros. As an example, a commonly used sparse prior for $\m{x}$ is a Laplace prior \cite{tibshirani1996regression,park2008bayesian}: $p\sbra{\m{x}}=\lambda^N\exp\lbra{-\lambda\onen{\m{x}}}$ with $\lambda$ being a positive constant. In such a case, we have $f\sbra{\m{x}}=\onen{\m{x}}$ that has been extensively studied in deterministic optimization methods.

\subsubsection{Quantization error model} We assume a uniform, noninformative prior for $\m{e}$:
\equ{\m{e}\sim U\sbra{\cD_e} \label{formu:prior_e}}
since the only information of $\m{e}$ that we use is $\m{e}\in\cD_e$.

\begin{rem} The uniform prior may not characterize well the quantization error in the case of a very small bit depth $B$. But it is noted that a sophisticated prior needs more information besides the domain $\cD_e$ which is difficult to obtain.
\end{rem}

To obtain an MAP estimator of $\m{x}$ requires to integrate out $\m{e}$ from $p\sbra{\m{z},\m{x},\m{e}}$ that is computationally intractable. We propose to estimate $\m{x}$ and $\m{e}$ simultaneously using their joint MAP estimator:
\equ{\begin{split}\lbra{\widehat{\m{x}},\widehat{\m{e}}}
&=\arg\max_{\m{x},\m{e}} \log p\sbra{\m{x},\m{e}|\m{z}} \\
&=\arg\max_{\m{x},\m{e}} \log p\sbra{\m{z},\m{x},\m{e}} \\
&=\arg\max_{\m{x},\m{e}} \log\lbra{p\sbra{\m{z}|\m{x},\m{e}}p\sbra{\m{x}}p\sbra{\m{e}}}\\
&=\arg\min_{\m{x},\m{e}\in\cD_e} \lbra{f\sbra{\m{x}}+\frac{C_1}{2\sigma^2}\twon{\m{z}-\m{e}-\m{A}\m{x}}^2}. \end{split} \label{formu:MAPestimator}}
An equivalent form of the problem in (\ref{formu:MAPestimator}) is
\equ{\min_{\tilde{\m{x}},\tilde{\m{e}}}f\sbra{\tilde{\m{x}}},\st\left\{\begin{array}{l} \twon{\m{z}-\tilde{\m{e}}-\m{A}\tilde{\m{x}}}\leq\epsilon,\\\tilde{\m{e}}\in\cD_e,\end{array}\right. \label{formu:uniformoptimization}}
where $\epsilon$ is a proper scalar that controls the noise energy. The first constraint in (\ref{formu:uniformoptimization}) is to ensure the data consistency against the measurement noise. In multi-bit CS, the second one concerns data consistency due to quantization. In 1-bit CS, an additional signal scaling constraint is included in the second constraint that prevents an optimal solution for $\m{x}$ from $\m{0}$. Before proceeding to our algorithm within the framework of Bayesian CS, we study in the next subsection relations of the proposed Bayesian framework with existing methods.

\subsection{Relations with Existing Methods in Quantized CS}
We first note that problem (\ref{formu:uniformoptimization}) is equivalent to the problem
\equ{\min_{\tilde{\m{x}},\tilde{\m{y}}}f\sbra{\tilde{\m{x}}},\st\left\{\begin{array}{l} \twon{\tilde{\m{y}}-\m{A}\tilde{\m{x}}}\leq\epsilon,\\\tilde{\m{y}}\in\cD_y.\end{array}\right. \label{formu:opt_multi-bit}}
In the following we show that many existing problem formulations of quantized CS are special cases of or related to (\ref{formu:opt_multi-bit}).

\subsubsection{Multi-bit CS}
We consider the case of $\ell_1$ optimization where $f\sbra{\m{x}}=\onen{\m{x}}$. In the noise free case where $\epsilon=0$, the problem in (\ref{formu:opt_multi-bit}) can be written into
\equ{\min_{\tilde{\m{x}}}\onen{\tilde{\m{x}}},\st\m{A}\tilde{\m{x}}\in\cD_y, \label{formu:opt_multi-bit_noisefree}}
which has been studied in \cite{dai2011information}. Further, by assuming that $\cQ$ is a uniform unsaturated quantizer the above problem can be written into
\equ{\min_{\tilde{\m{x}}}\onen{\tilde{\m{x}}},\st\inftyn{\m{z}-\m{A}\tilde{\m{x}}}\leq \frac{r}{2} \label{formu:opt_multi-bit_noisefree_unifquan}}
which is studied in \cite{candes2006near,jacques2011dequantizing} with $r$ denoting the quantization bin width. While existing methods that account for measurement noise typically mix it up with the quantization error, e.g., in \cite{laska2011democracy}, problem (\ref{formu:opt_multi-bit}) extends existing noise free formulations to the noisy case by dealing with the two uncertainties separately.

\begin{rem} Under the assumption that all quantization errors are independent and uniformly distributed in a common interval $\mbra{-\frac{r}{2},\frac{r}{2}}$, it is shown in \cite{jacques2011dequantizing} that the $\ell_{\infty}$ norm in problem (\ref{formu:opt_multi-bit_noisefree_unifquan}) is not the best choice for the signal recovery. But it is unclear whether the result in \cite{jacques2011dequantizing} can be extended to the case of a general quantizer where the above assumption fails. It is noted that our problem formulation does not require this assumption and applies to an arbitrary quantizer. That is, by losing some optimality, we have obtained the universality.
\end{rem}

\subsubsection{1-bit CS} In 1-bit CS (\ref{formu:opt_multi-bit}) becomes
\equ{\min_{\tilde{\m{x}},\tilde{\m{y}}}f\sbra{\tilde{\m{x}}},\st\left\{\begin{array}{l} \twon{\tilde{\m{y}}-\m{A}\tilde{\m{x}}}\leq\epsilon,\\\sgn\sbra{\tilde{\m{y}}}=\sgn\sbra{\m{z}},\\ \norm{\tilde{\m{y}}}_s=1.\end{array}\right. \label{formu:opt_1-bit}}
In the noise free case, it can be written into
\equ{\min_{\tilde{\m{x}}}f\sbra{\tilde{\m{x}}},\st\left\{\begin{array}{l} \sgn\sbra{\m{A}\tilde{\m{x}}}=\sgn\sbra{\m{z}},\\ \norm{\m{A}\tilde{\m{x}}}_s=1.\end{array}\right. \label{formu:opt_1-bit_noisefree}}
This problem with the settings $f\sbra{\m{x}}=\onen{\m{x}}$ and $s=1$ can be shown to be convex and has been studied in \cite{plan2011one}. So by (\ref{formu:opt_1-bit}) we extend (\ref{formu:opt_1-bit_noisefree}) to the noisy case while the authors of \cite{plan2011one} state in \cite{plan2012robust} that ``it was unclear how to modify the above convex program to account for possible noise.''

The third constraint in (\ref{formu:opt_1-bit}) serves to prevent the optimal solution for $\m{x}$ from $\m{0}$. If replacing it by $\norm{\tilde{\m{x}}}_2=1$, then problem (\ref{formu:opt_1-bit}) can be shown to be equivalent to the problem
\equ{\min_{\tilde{\m{x}}}f\sbra{\tilde{\m{x}}},\st\left\{\begin{array}{l} \twon{\sbra{\sgn\sbra{\m{z}}\odot\m{A}\tilde{\m{x}}}_-}\leq\epsilon,\\ \norm{\tilde{\m{x}}}_2=1,\end{array}\right. \label{formu:opt_1-bit3}}
where $\sbra{v}_-=\max\lbra{-v,0}$ for a scalar $v$ and operates elementwise for a vector. In the noise free case (\ref{formu:opt_1-bit3}) becomes
\equ{\min_{\tilde{\m{x}}}f\sbra{\tilde{\m{x}}},\st\left\{\begin{array}{l} \sgn\sbra{\m{A}\tilde{\m{x}}}=\sgn\sbra{\m{z}},\\ \norm{\tilde{\m{x}}}_2=1,\end{array}\right. \label{formu:opt_1-bit_noisefree2}}
which with $f\sbra{\m{x}}=\onen{\m{x}}$ is the earliest formulation of the 1-bit CS problem introduced in \cite{boufounos2008bit} and solved using RFPI in \cite{boufounos2008bit} and RSS in \cite{laska2011trust}. Assume that the signal sparsity information is known instead of the noise energy, another formulation of (\ref{formu:opt_1-bit3}) is to pose $\twon{\sbra{\sgn\sbra{\m{z}}\odot\m{A}\tilde{\m{x}}}_-}$ as the objective function and $f\sbra{\tilde{\m{x}}}\leq S$ as a constraint, where the constant $S$ refers to the sparsity information. Such kind of formulations have been studied in \cite{boufounos2008bit,boufounos2009greedy, jacques2011robust}. In addition, the convex program in \cite{plan2012robust} is related by observing that $\sgn^T\sbra{\m{z}}\m{A}\tilde{\m{x}}= \onen{\sbra{\sgn\sbra{\m{z}}\odot\m{A}\tilde{\m{x}}}_+} -\onen{\sbra{\sgn\sbra{\m{z}}\odot\m{A}\tilde{\m{x}}}_-}$, where $(v)_+=\max\lbra{v,0}$.

\begin{rem} By (\ref{formu:opt_1-bit3}) we see that the effective noise is $\sbra{\sgn\sbra{\m{z}}\odot\m{A}\m{x}}_-$ in 1-bit CS where by ``effective noise'' we refer to a noise that has the minimum energy and leads to the same measurement. Since its energy is much smaller than that of the true noise $\m{y}-\m{A}\m{x}$, from this point of view, we may say that 1-bit CS is robust to the measurement noise. \label{rem:robust_1bit}
\end{rem}

\section{Q-VMP: Variational Message Passing for Quantized CS} \label{sec:QVMP}

\subsection{Model Selection}
We assume that the noise variance $\sigma^2$ is known. Though it can be estimated by assuming an inverse Gamma prior for it in the case where it is unknown as in \cite{babacan2010bayesian,yang2013off}, its estimate is inaccurate due to an ``identifiability issue'' as addressed in \cite{wipf2007empirical}. For the sparse signal $\m{x}$, we adopt a three-layer, Gaussian-Gamma-Gamma, hierarchical prior  introduced in \cite{pedersen2011sparse}:
\[p\sbra{\m{x};\epsilon,c,d}=\iint p\sbra{\m{x}|\m{\alpha}}p\sbra{\m{\alpha}|\eta;\epsilon}p\sbra{\eta;c,d}\,d\m{\alpha}\,d\eta\]
where
{\lentwo\equa{p\sbra{\m{x}|\m{\alpha}}
&=&\cN\sbra{\m{x}|\m{0},\m{\Lambda}},\label{formu:prior_x}\\p\sbra{\m{\alpha}|\eta;\epsilon}
&=&\prod_{i=1}^{N}\Gamma\sbra{\alpha_i|\epsilon,\eta},\label{formu:prior_alpha}\\p\sbra{\eta;c,d}
&=&\Gamma\sbra{\eta|c,d}\label{formu:prior_eta}
}}with $\m{\Lambda}=\diag\sbra{\m{\alpha}}$ and constants $\epsilon$, $c$, $d$. For a Gamma distributed variable $u\sim\Gamma\sbra{c,d}$, its PDF is
$\Gamma\sbra{u|c,d}=\frac{d^c}{\Gamma\sbra{c}}u^{c-1}\exp\sbra{-du}$
with $\Gamma\sbra{c}$ being the Gamma function. By \cite{pedersen2011sparse} the constants $\epsilon$, $c$, $d$ satisfy that $0\leq\epsilon\leq1$, $c,d\geq0$. In this paper, we adopt $c=1$, $d=0$ to make the prior for $\eta$ in (\ref{formu:prior_eta}) noninformative (flat on $\bR_+$). Further, we choose $\epsilon=0$ since a smaller $\epsilon$ leads to a sparser prior and an estimator that approximates a hard-thresholding rule according to \cite{pedersen2011sparse}. Readers are referred to \cite{pedersen2011sparse} for more properties of the Gaussian-Gamma-Gamma prior and its relations with other sparse estimation techniques.

In 1-bit CS, we let $\m{y}$ have unit $\ell_2$ norm in (\ref{formu:norm_y_1bit}), leading to that $\twon{\m{e}}=1$ in (\ref{formu:domain_e_1bit}). As a result, we have the joint PDF of the observation model (\ref{formu:ObservationModel}):
\equ{p\sbra{\m{z},\m{x},\m{e},\m{\alpha},\eta} = p\sbra{\m{z}|\m{x},\m{e}}p\sbra{\m{x}|\m{\alpha}}p\sbra{\m{\alpha}|\eta} p\sbra{\eta}p\sbra{\m{e}} \label{formu:jointpdf}}
with the distributions on the right hand side as defined respectively by (\ref{formu:prior_n}), (\ref{formu:prior_x}), (\ref{formu:prior_alpha}), (\ref{formu:prior_eta}) and (\ref{formu:prior_e}). A directed graphical model that encodes the factorization of the joint PDF in (\ref{formu:jointpdf}) is shown in Fig. \ref{Fig:factorgraph}.

\begin{figure}
  \centering
  \includegraphics[width=3in]{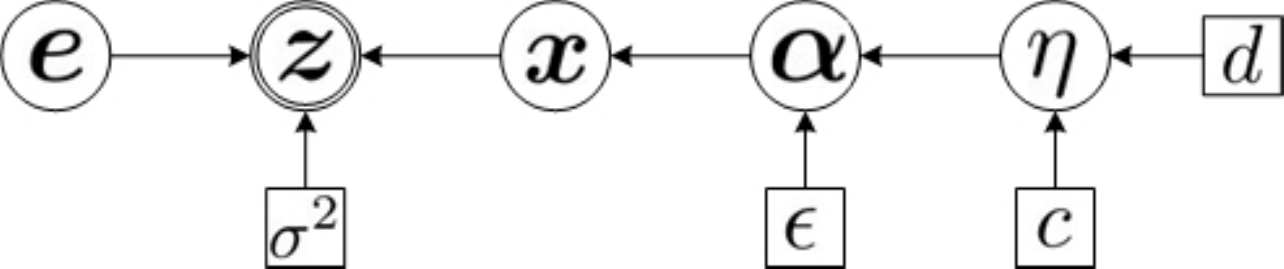}
  \caption{Directed graphical model that encodes the joint PDF in (\ref{formu:jointpdf}) of the Bayesian model. Nodes denoted with circles correspond to random variables, while nodes denoted with squares correspond to parameters of the model. Doubly circled $\m{z}$ is the observation while single circled nodes represent hidden variables.}
  \label{Fig:factorgraph}
\end{figure}

\subsection{Q-VMP Algorithm}\label{sec:Bayesianinference}
It is known that Bayesian inference is based on the posterior distribution
$p\sbra{\m{x},\m{e},\m{\alpha},\eta|\m{z}} = p\sbra{\m{z},\m{x},\m{e},\m{\alpha},\eta}/p\sbra{\m{z}}$.
However, such an exact posterior distribution is intractable since
$p\sbra{\m{z}}=\idotsint p\sbra{\m{z},\m{x},\m{e},\m{\alpha},\eta}\,d\m{x}\,d\m{e}\,d\m{\alpha}\,d\eta$
cannot be expressed explicitly.

A variational inference approach \cite{jordan1999introduction,tzikas2008variational} is adopted in this paper. Denote $\m{V}=\lbra{\m{x},\m{e},\m{\alpha},\eta}$ the set of all unknown variables to be estimated. The goal in variational inference is to find a tractable distribution $q\sbra{\m{V}}$ that closely approximates the true posterior distribution $p\sbra{\m{V}|\m{z}}$. To do this, some family of distributions that has enough flexibility is firstly chosen to represent $q\sbra{\m{V}}$. Then the task is to find a member of the family that minimizes the Kullback-Leibler (KL) divergence between the true posterior $p\sbra{\m{V}|\m{z}}$ and the variational approximation $q\sbra{\m{V}}$. A commonly used variational distribution $q\sbra{\m{V}}$ is such that disjoint groups of variables are independent, i.e., $q\sbra{\m{V}}$ has a factorized form
$q\sbra{\m{V}}=q\sbra{\m{x}}q\sbra{\m{e}}q\sbra{\m{\alpha}} q\sbra{\eta}$.
Variational message passing (VMP) is proposed in \cite{winn2006variational} for the variational inference using a message passing procedure on a graphical model. In VMP, the variational distributions $q\sbra{\m{x}}$, $q\sbra{\m{e}}$, $q\sbra{\m{\alpha}}$, $q\sbra{\eta}$ are iteratively updated to monotonically decrease the KL divergence and thus has guaranteed convergence. Readers are referred to \cite{winn2006variational} for more details of VMP.
The updates of $q\sbra{\m{x}}$, $q\sbra{\m{\alpha}}$, $q\sbra{\eta}$ are similar to those in \cite{pedersen2011sparse} because of the similarity between quantized and conventional CS. $q\sbra{\m{e}}$ is given complete flexibility in multi-bit CS as $q\sbra{\m{x}}$, $q\sbra{\m{\alpha}}$ and $q\sbra{\eta}$. We constrain $q\sbra{\m{e}}$ in 1-bit CS such that \equ{q\sbra{\m{e}}=\delta\sbra{\m{e}-\m{e}^0} \label{formu:posterior_e_1bit}}
due to a computational issue to be discussed in Remark \ref{rem:whyconstaine1-bit}, where $\delta\sbra{\cdot}$ is the delta function and $\m{e}^0\in\bR^M$ is to be estimated. Note that (\ref{formu:posterior_e_1bit}) is equivalent to the complete flexibility in the noise free case to be illustrated in Subsection \ref{sec:noisefreecase}.

\begin{rem} In the 1-bit case the convergence of the resulting algorithm is not a direct result of \cite{winn2006variational} due to the adoption of a degenerate distribution for $\m{e}$. Instead, we may consider $\m{e}$ as an unknown deterministic parameter in such a case. Then the resulting algorithm can be interpreted as a variational EM algorithm \cite{tzikas2008variational} and during the iterations the (marginal) likelihood $p\sbra{\m{z};\m{e}}$ is guaranteed to monotonically increase and thus convergence is guaranteed. Readers are referred to \cite{tzikas2008variational} for the details.
\end{rem}

\subsubsection{Updates of $q\sbra{\m{x}}$, $q\sbra{\m{\alpha}}$ and $q\sbra{\eta}$}
According to \cite{winn2006variational} we have that
\[\begin{split}q\sbra{\m{x}}
&\propto \exp\lbra{\inp{\ln p\sbra{\m{z}|\m{x},\m{e}}}_{q\sbra{\m{e}}}\inp{\ln p\sbra{\m{x}|\m{\alpha}}}_{q\sbra{\alpha}}} \\
&\propto \exp\lbra{-\frac{1}{2}\sbra{\m{x}-\m{\mu}}^T\m{\Sigma}^{-1}\sbra{\m{x}-\m{\mu}}}, \end{split}\]
and thus $q\sbra{\m{x}}$ is a Gaussian distribution $\cN\sbra{\m{x}|\m{\mu},\m{\Sigma}}$ with the mean $\m{\mu}$ and covariance $\m{\Sigma}$:
{\lentwo\equa{ \m{\mu}
&=& \inp{\m{x}}_{q\sbra{\m{x}}} =\sigma^{-2}\m{\Sigma}\m{A}^T \sbra{\m{z}-\inp{\m{e}}_{q\sbra{\m{e}}}}, \label{formu:update_mu}\\ \m{\Sigma}
&=& \sbra{\sigma^{-2}\m{A}^T\m{A} +\inp{\m{\Lambda}^{-1}}_{q\sbra{\m{\alpha}}}}^{-1}. \label{formu:update_Sigma}}
}For $\m{\alpha}$ we have
\[\begin{split} q\sbra{\m{\alpha}}
&\propto \exp\lbra{\inp{\ln p\sbra{\m{x}|\m{\alpha}}}_{q\sbra{\m{x}}} \inp{\ln p\sbra{\m{\alpha|\eta}}}_{q\sbra{\eta}}} \\
&\propto\prod_{n=1}^N \alpha_n^{\epsilon-\frac{3}{2}} \exp\lbra{-\frac{1}{2}\alpha_n^{-1}\inp{x_n^2}_{q\sbra{\m{x}}} -\alpha_n\inp{\eta}_{q\sbra{\eta}}} \end{split}\]
where $\inp{x_n^2}_{q\sbra{\m{x}}}=\mu_n^2+\Sigma_{nn}$. The expression on the right hand side is the product of generalized inverse Gaussian (GIG) PDFs and thus we have for any $i\in\bR$ \cite{jorgensen1982statistical}:
\equ{\inp{\alpha_n^i}_{q\sbra{\m{\alpha}}} = \sbra{\frac{\inp{x_n^2}_{q\sbra{\m{x}}}} {2\inp{\eta}_{q\sbra{\eta}}}}^{\frac{i}{2}}\frac{\cK_{\epsilon+i-\frac{1}{2}} \sbra{\sqrt{2\inp{\eta}_{q\sbra{\eta}} \inp{x_n^2}_{q\sbra{\m{x}}}}}} {\cK_{\epsilon-\frac{1}{2}} \sbra{\sqrt{2\inp{\eta}_{q\sbra{\eta}} \inp{x_n^2}_{q\sbra{\m{x}}}}}} \label{formu:update_alphamoments}}
where $\cK_{\nu}\sbra{\cdot}$ is the modified Bessel function of the second kind and order $\nu\in\bR$. The case of $i=-1$ in (\ref{formu:update_alphamoments}) gives the evaluation of $\inp{\m{\Lambda}^{-1}}_{q\sbra{\m{\alpha}}}$ used in (\ref{formu:update_Sigma}), and the case of $i=1$ gives the calculation of $\inp{\alpha_n}_{q\sbra{\m{\alpha}}}$ used in a later expression in (\ref{formu:update_eta}). The update of $q\sbra{\eta}$ can be shown to be
$q\sbra{\eta}=\Gamma\sbra{\eta|N\epsilon+c,\sum_{n=1}^N \inp{\alpha_n}_{q\sbra{\m{\alpha}}} +d}$.
The first moment of $\eta$ used in (\ref{formu:update_alphamoments}) is given as \equ{\inp{\eta}_{q\sbra{\eta}} = \frac{N\epsilon+c} {\sum_{n=1}^N \inp{\alpha_n}_{q\sbra{\m{\alpha}}} +d}. \label{formu:update_eta}}

\subsubsection{Update of $q\sbra{\m{e}}$ in multi-bit CS}
In multi-bit CS we have
\equ{\begin{split} q\sbra{\m{e}}
&\propto \exp\lbra{\inp{\ln p\sbra{\m{z}|\m{x},\m{e}}}_{q\sbra{\m{x}}}} p\sbra{\m{e}}\\
&\propto \exp\lbra{-\frac{1}{2}\sigma^{-2} \inp{\twon{\m{z}-\m{e}-\m{A}\m{x}}^2}_{q\sbra{\m{x}}}}I_{\m{e}}\sbra{\cD_e}\\
&\propto \exp\lbra{-\frac{1}{2}\sigma^{-2} \twon{\m{e}-\sbra{\m{z}-\m{A}\m{\mu}}}^2} I_{\m{e}}\sbra{\cD_e}, \end{split} \label{formu:posterior_e}}
where $I_{\m{e}}\sbra{\cD_e}$ is an indicator function that equals to 1 if $\m{e}\in\cD_e$ or 0 otherwise. Hence, $q\sbra{\m{e}}$ is the product of PDFs of truncated Gaussian distributions, i.e., for each $m=1,\cdots,M$, $q\sbra{e_m}$ is the PDF of a truncated Gaussian distribution. As a result, the first moment of $e_m$, $m=1,\cdots,M$, used in (\ref{formu:update_mu}) can be given in closed form after some derivations using the PDF $\phi\sbra{\cdot}$ and cumulative distribution function (CDF) $\Phi\sbra{\cdot}$ of a standard Gaussian distribution:
\equ{\inp{e_m}_{q\sbra{\m{e}}}
= \sigma\frac{\phi\sbra{l_{e_m}}-\phi\sbra{u_{e_m}}} {\Phi\sbra{u_{e_m}}-\Phi\sbra{l_{e_m}}} + \mu_{e_m}, \label{formu:update_em1stmoment}}
where $\mu_{e_m}=\sbra{\m{z}-\m{A}\m{\mu}}_m$, $l_{e_m}$ and $u_{e_m}$ satisfy that $\cD_{e_m}=\mbra{\sigma l_{e_m}+\mu_{e_m}, \sigma u_{e_m}+\mu_{e_m}}$ with $\cD_{e_m}$ denoting the domain of $e_m$, $\phi\sbra{u}=\frac{1}{\sqrt{2\pi}}\exp\lbra{-\frac{u^2}{2}}$ and $\Phi\sbra{u}=\int_{-\infty}^u\phi\sbra{t}\,dt$ for $u\in\bR$.

\begin{rem}
Consider the case where $q\sbra{\m{e}}$ is given the complete flexibility in 1-bit CS. Note that entries of a point in $\cD_e$ are no longer independent of each other in such a case, leading to that $q\sbra{\m{e}}$ is the PDF of a truncated multi-variable Gaussian distribution with $\m{e}$ constrained in a nonconvex set $\cD_e$ defined in (\ref{formu:domain_e_1bit}). As a result, the calculation of $\inp{\m{e}}_{q\sbra{\m{e}}}$ is in general computationally intractable in our considered CS problems where the dimension of $\m{e}$ is large. \label{rem:whyconstaine1-bit}
\end{rem}

\subsubsection{Update of $q\sbra{\m{e}}$ in 1-bit CS}
According to \cite{winn2006variational}, this is equivalent to finding an MAP estimator of $\m{e}$ with its posterior distribution defined in (\ref{formu:posterior_e}). So we have
\equ{\begin{split} \inp{\m{e}}_{q\sbra{\m{e}}}
&= \arg\max_{\m{e}\in\cD_e} \exp\lbra{-\frac{1}{2}\sigma^{-2} \twon{\m{e}-\sbra{\m{z}-\m{A}\m{\mu}}}^2}\\
&= \cP_{\cD_e}\sbra{\m{z}-\m{A}\m{\mu}}, \end{split}  \label{formu:update_e_1bit}}
where $\cD_e$ is defined in (\ref{formu:domain_e_1bit}) and $\cP_{\cD}\sbra{\m{v}}$ denotes a projection of a point $\m{v}$ onto a set $\cD$. The calculation of $\cP_{\cD_e}\sbra{\cdot}$ with the nonconvex set $\cD_e$ is provided in the following lemma.

\begin{lem} For a vector $\m{v}\in\bR^M$, let $\overline{\m{v}}=-\sgn\sbra{\m{z}}\odot\m{v}$. Denote $\cI$ the index set of all positive entries of $\overline{\m{v}}$. Let $\cI^c$ be its complementary set. If $\cI$ is nonempty, then let $\m{e}^*\in\bR^M$ with $\m{e}^*_{\cI}= \frac{\m{v}_{\cI}}{\twon{\m{v}_{\cI}}}$ and $\m{e}^*_{\cI^c}= \m{0}$. Otherwise, let $i_0=\arg\max_i\sbra{\overline{v}_i}$ and $\m{e}^*$ such that $e_{i_0}^*=-\sgn\sbra{z_{i_0}}$ and $e_i^*=0$ whenever $i\neq i_0$. Then $\m{e}^*=\cP_{\cD_e}\sbra{\m{v}}$ with $\cD_e$ as defined in (\ref{formu:domain_e_1bit}). \label{lem:calculateprojection}
\end{lem}
\begin{proof} See Appendix.
\end{proof}

\begin{figure}
  \centering
  \includegraphics[width=2in]{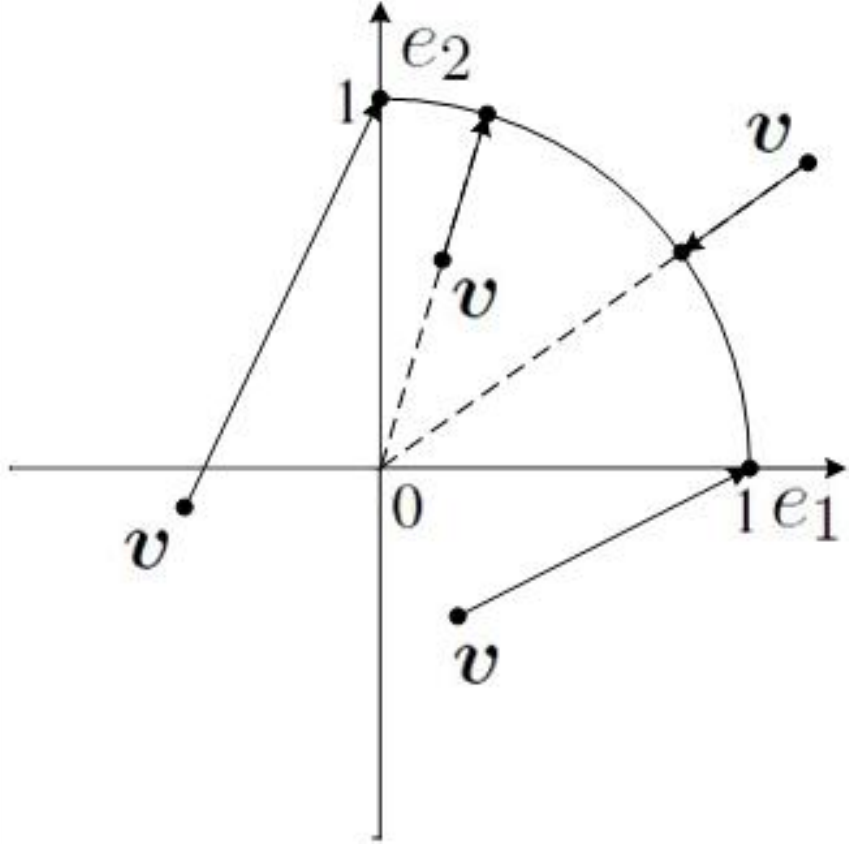}
  \caption{An illustration of Lemma \ref{lem:calculateprojection} with nonnegative entries of $\m{e}$. The unit circle in the first quadrant composes of $\cD_e$. Projections of four possible $\m{v}$'s are shown.}
  \label{Fig:ProjectionIllustration}
\end{figure}

Lemma \ref{lem:calculateprojection} tells how to calculate the projection onto the nonconvex set $\cD_e$ defined in (\ref{formu:domain_e_1bit}). An illustration of Lemma \ref{lem:calculateprojection} is presented in Fig. \ref{Fig:ProjectionIllustration}, where we consider the two dimensional case with both entries of $\m{e}$ nonnegative. The unit circle in the first quadrant composes of $\cD_e$. Projections of four possible $\m{v}$'s are shown. The resulting algorithm is summarized in Algorithm \ref{alg:Q-VMP}, named as variational message passing with quantization (Q-VMP).

\begin{algorithm}
\caption{Q-VMP}
Input: sensing matrix $\m{A}$, quantized measurement $\m{z}$, domain of quantization error $\cD_e$, and noise variance $\sigma^2$.\\
1. initialize $\inp{\alpha_n^{-1}}_{q\sbra{\m{\alpha}}}$, $n=1,\cdots,N$, $\inp{\eta}_{q\sbra{\eta}}$ and $\inp{\m{e}}_{q\sbra{\m{e}}}$;\\
2. \textbf{while} not converged \textbf{do} \\
3. $\quad$ update $\m{\Sigma}$ by (\ref{formu:update_Sigma});\\
4. $\quad$ update  $\m{\mu}$ by (\ref{formu:update_mu});\\
5. $\quad$ update $\inp{\alpha_n^{-1}}_{q\sbra{\m{\alpha}}}$ and $\inp{\alpha_n}_{q\sbra{\m{\alpha}}}$, $n=1,\cdots,N$, by (\ref{formu:update_alphamoments});\\
6. $\quad$ update $\inp{\eta}_{q\sbra{\eta}}$ by (\ref{formu:update_eta});\\
7. $\quad$ update $\inp{\m{e}}_{q\sbra{\m{e}}}$ by (\ref{formu:update_em1stmoment}) in multi-bit CS or by (\ref{formu:update_e_1bit}) in 1-bit CS; \\
8. \textbf{end} while\\
Output: recovered signal $\widehat{\m{x}}=\m{\mu}$. \label{alg:Q-VMP}
\end{algorithm}


\subsection{The Noise Free Case} \label{sec:noisefreecase}
In this subsection we consider Q-VMP in the noise free case. We first consider the data consistency. A consistent recovery means that the observation can be reproduced from the recovered signal. Empirical results suggest that a consistent recovery result in less errors \cite{jacques2011dequantizing,laska2011trust}. A theoretical proof is provided in \cite{jacques2011robust} on the 1-bit case. The following analysis applies to both multi- and 1-bit CS. Taking $\sigma^2\rightarrow0$ at both sides of (\ref{formu:update_mu}) and (\ref{formu:update_Sigma}) gives
\[\begin{split}\m{\mu}
&\rightarrow \overline{\m{\Lambda}}^{\frac{1}{2}} \sbra{\m{A}\overline{\m{\Lambda}}^{\frac{1}{2}}}^\dag \sbra{\m{z}-\m{e}},\\ \m{\Sigma}
&\rightarrow\overline{\m{\Lambda}}- \overline{\m{\Lambda}}^{\frac{1}{2}} \sbra{\m{A}\overline{\m{\Lambda}}^{\frac{1}{2}}}^\dag \m{A}\overline{\m{\Lambda}}, \end{split}\]
where $\overline{\m{\Lambda}}=\inp{\m{\Lambda}^{-1}}_{q\sbra{\m{\alpha}}}^{-1}$. Thus we have
$\m{A}\m{\mu}\rightarrow\m{z}-\m{e}\in\cD_y$, i.e., $\cQ\sbra{\m{A}\m{\mu}}\rightarrow\m{z}$, which indicates that the recovered signal reproduces the observation at each iteration.

We next consider the update of $q\sbra{\m{e}}$ in such a case. As $\sigma^2\rightarrow0$ we see that $q\sbra{\m{e}}$ degenerates into a single-point distribution by (\ref{formu:posterior_e}), that coincides with the stricter assumption in (\ref{formu:posterior_e_1bit}) in 1-bit CS.

\subsection{Pruning a Basis Function}
The most difficult computation of Q-VMP is the calculation of $\m{\Sigma}$ that is the inverse of an $N\times N$ matrix. Using the Woodbury matrix identity, we have
\[\m{\Sigma}=\inp{\m{\Lambda}^{-1}}_{q\sbra{\m{\alpha}}}^{-1} - \inp{\m{\Lambda}^{-1}}_{q\sbra{\m{\alpha}}}^{-1}\m{A}^T\m{C}^{-1}\m{A} \inp{\m{\Lambda}^{-1}}_{q\sbra{\m{\alpha}}}^{-1} \]
with $\m{C}=\sigma^{2}\m{I} +\m{A}\inp{\m{\Lambda}^{-1}}_{q\sbra{\m{\alpha}}}^{-1}\m{A}^T$ being an $M\times M$ matrix.
Hence, to calculate $\m{\Sigma}$ needs $O\sbra{\min\lbra{N^3,N^2M}}$ operations. It is noted that if Q-VMP produces some $\inp{\alpha_n^{-1}}_{q\sbra{\m{\alpha}}}\rightarrow+\infty$ with $n\in\lbra{1,\cdots,N}$, then the corresponding basis $\m{A}_n$ can be removed from the model. To further speed up Q-VMP, we prune a basis $\m{A}_n$ from the model (to reduce $N$) when the corresponding parameter $\inp{\alpha_n^{-1}}_{q\sbra{\m{\alpha}}}$ is larger than a certain threshold $\tau_{pruning}$. Similar basis pruning approaches have been used in \cite{tipping2001sparse,pedersen2011sparse}.

\section{Numerical Simulations} \label{sec:simulation}
In this section, we study the performance of the proposed observation model and Q-VMP algorithm in comparison with existing ones by numerical simulations.

\subsection{Experimental Setup}
\textbf{Quantizer:} In multi-bit CS, a uniform unsaturated quantizer is defined in (\ref{formu:quantizer}) with $L=2^B$, equispaced $u_0,u_1\cdots,u_{L}$ and $v_i=\sbra{u_i+u_{i+1}}/2$, $i=1,\cdots,L-1$. In addition, we let $u_L=\inftyn{\m{y}}$ and $u_0=-\inftyn{\m{y}}$ in each trial. For a saturated quantizer, we set $u_0=-\infty$, $u_L=+\infty$.

\textbf{CS problem generation:} In our experiment, we set $N=500$, $K=10$, and vary the bit budget (total bits of all measurements) in $\lbra{50,100,\cdots,1000}$. In each trial, a $K$-sparse signal of length $N$ is generated with Gaussian distributed nonzero entries and then scaled to unit norm. Entries of the sensing matrix $\m{A}$ are generated independently according to a Gaussian distribution $\cN\sbra{0,M^{-1}}$. Thus the noise free measurement $\m{y}^0=\m{A}\m{x}$ has unit norm in expectation. To obtain a desired SNR, a white Gaussian measurement noise $\m{n}$ is added with the noise variance
$\sigma^2=M^{-1}10^{-\frac{\text{SNR}}{10}}$. The quantized measurement $\m{z}=\cQ\sbra{\m{y}}$ is preserved for the following signal recovery.

\textbf{Performance metrics:} Three metrics are considered, including reconstruction SNR (RSNR), sparsity level of the recovered signal and computational speed. RSNR is defined as
$\text{RSNR}=-20\log_{10}\twon{\m{x}-\widehat{\m{x}}}$,
where $\widehat{\m{x}}$ denotes the recovered signal of $\m{x}$. The sparsity level is measured by the support size of the recovered signal. The computational speed is measured by the CPU time usage. All results are averaged over 200 trials.

\subsection{Model Efficiency}
We first study the efficiency of the observation model in (\ref{formu:ObservationModel}) introduced in this paper for quantized CS. We consider the multi-bit CS problem with a uniform unsaturated quantizer as an example. In existing methods that account for measurement noise, e.g., in \cite{laska2011democracy}, the quantization error and the noise are typically coupled and treated as a Gaussian noise (only the energy information is used). Then the quantized CS problem is transformed into a conventional one. We refer to this formulation as existing method hereafter. In this subsection we compare the signal recovery performance of the proposed formulation in (\ref{formu:ObservationModel}) with the existing one. Naturally, we use the proposed Q-VMP algorithm for our formulation. A corresponding algorithm for the existing formulation is thus VMP introduced in \cite{pedersen2011sparse} for conventional CS. The latter algorithm can be considered as a simplified version of Q-VMP with the quantization error $\m{e}$ fixed throughout the algorithm. In addition, we also present the performance of oracle-aided conventional CS in which the true-valued measurements are used and thus whose performance acts as an upper boundary of the quantized CS problem.

In our experiment, we set $\text{SNR}=30$dB and the bit depth $B=4$ which leads to the number of quantized measurements varying from 12 to 250. In Q-VMP, we initialize $\inp{\alpha_n^{-1}}_{q\sbra{\m{\alpha}}}=1/\abs{\m{A}_n^T\m{z}}$, $n=1,\cdots,N$, $\inp{\eta}_{q\sbra{\eta}}=1$ and $\inp{\m{e}}_{q\sbra{\m{e}}}=\m{0}$. We set $\tau_{pruning}=10^{4}$. Q-VMP is terminated if $\frac{\twon{\tilde{\m{\alpha}}^j-\tilde{\m{\alpha}}^{j-1}}}{\twon{\tilde{\m{\alpha}}^{j-1}}}<10^{-5}$ or the maximum number of iterations, set to 2000, is reached, where $\tilde{\m{\alpha}}=\mbra{\inp{\alpha_1^{-1}}_{q\sbra{\m{\alpha}}}^{-1}, \cdots, \inp{\alpha_N^{-1}}_{q\sbra{\m{\alpha}}}^{-1}}^T$ and the superscript $j$ indicates the iteration. The VMP algorithm for the other two cases is similarly implemented. The true noise variance is used in Q-VMP and conventional CS. For VMP with the existing formulation, we set the noise variance to ${r^2}/{12}+\sigma^2$ where $r=u_1-u_0$ denotes the quantization bin width. This value corresponds to a Gaussian noise whose energy is comparable with that of $\m{e}+\m{n}$ under the assumption that $\m{e}$ is uniformly distributed and independent of $\m{n}$.

Reconstruction SNRs of the three methods are depicted in Fig. \ref{Fig:RSNR_model_efficiency}. It can be seen that the VMP algorithm based on the proposed observation model is consistently better than that with the existing formulation though it is worse than the oracle-aided one. So it confirms that the proposed model and framework improve the signal recovery accuracy by decoupling the quantization error from measurement noise.

\begin{figure}
\centering
\includegraphics[width=3.5in]{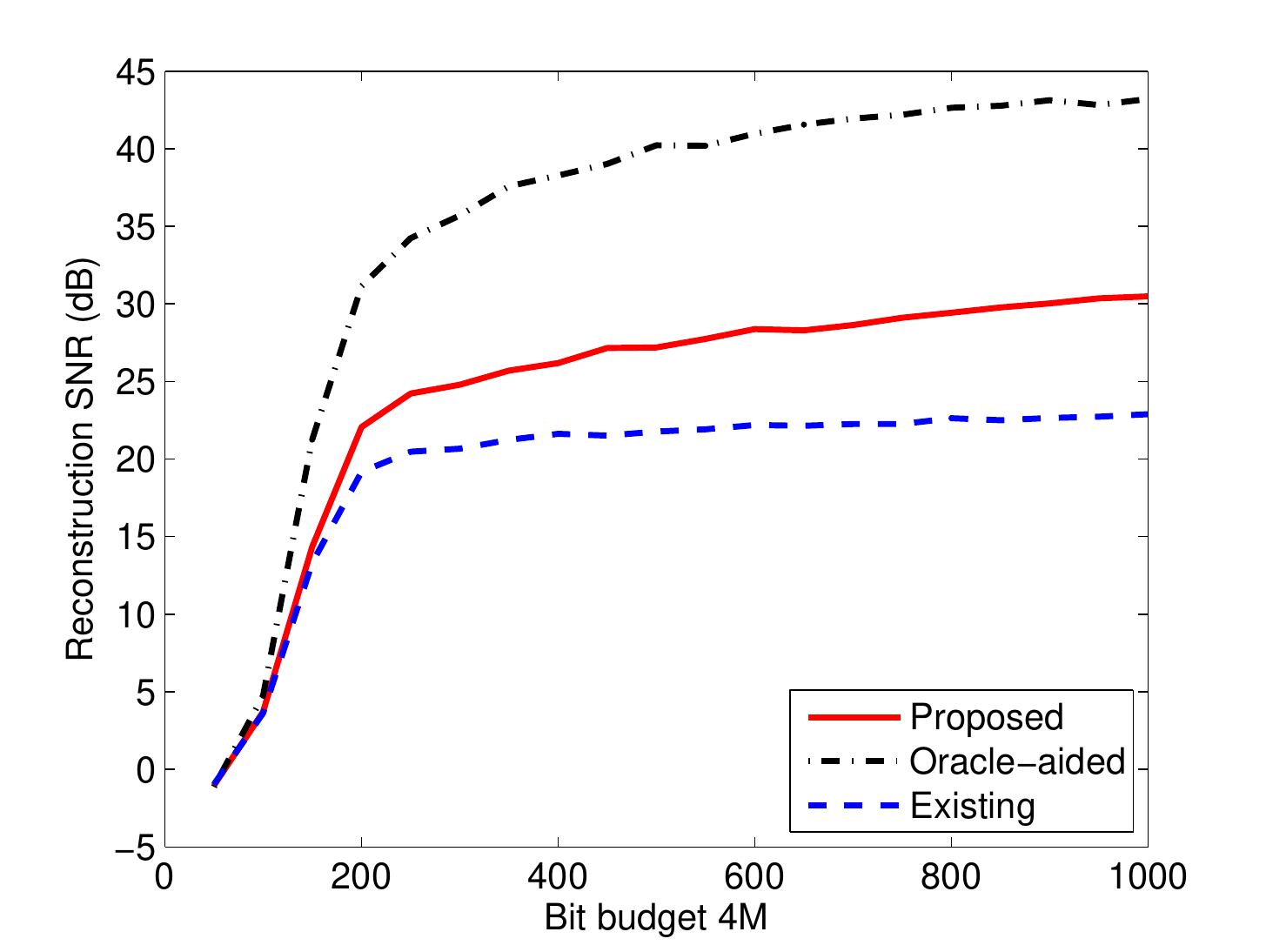}
\centering
\caption{Reconstruction SNRs of VMP algorithms implemented respectively based on the proposed observation model in (\ref{formu:ObservationModel}), an existing one that couples the quantization error and measurement noise, and conventional CS (oracle-aided quantized CS) as an upper boundary.}
\label{Fig:RSNR_model_efficiency}
\end{figure}

\subsection{Performance Comparison in Multi-bit CS} \label{sec:simulation_multibit}

\subsubsection{Unsaturated quantizer}

\begin{figure*}
\centering
  \subfigure[]{
    \label{Fig:RSNR_multibit_SNR30}
    \includegraphics[width=2.33in]{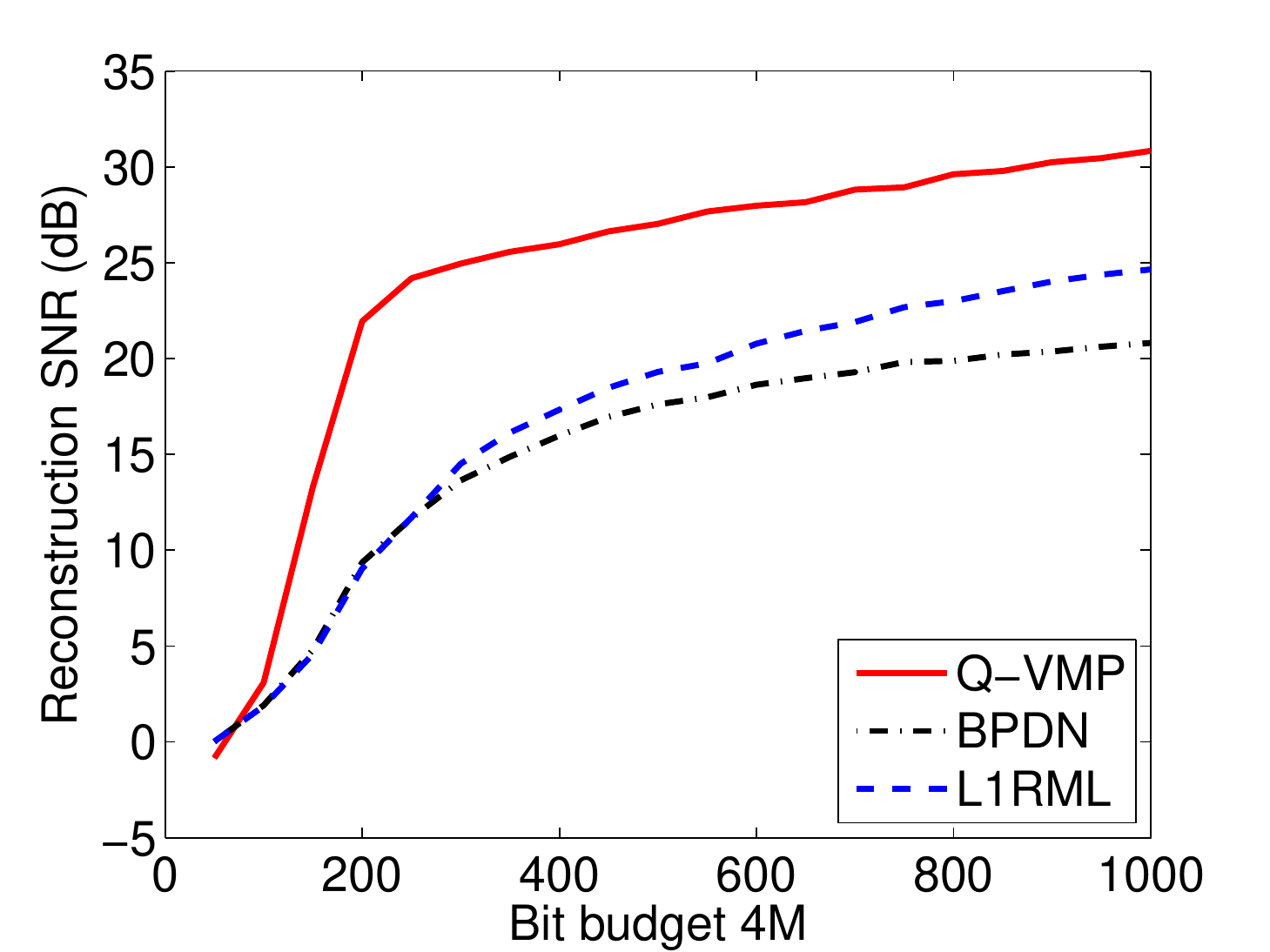}}%
  \subfigure[]{
    \label{Fig:SuppSize_multibit_SNR30}
    \includegraphics[width=2.33in]{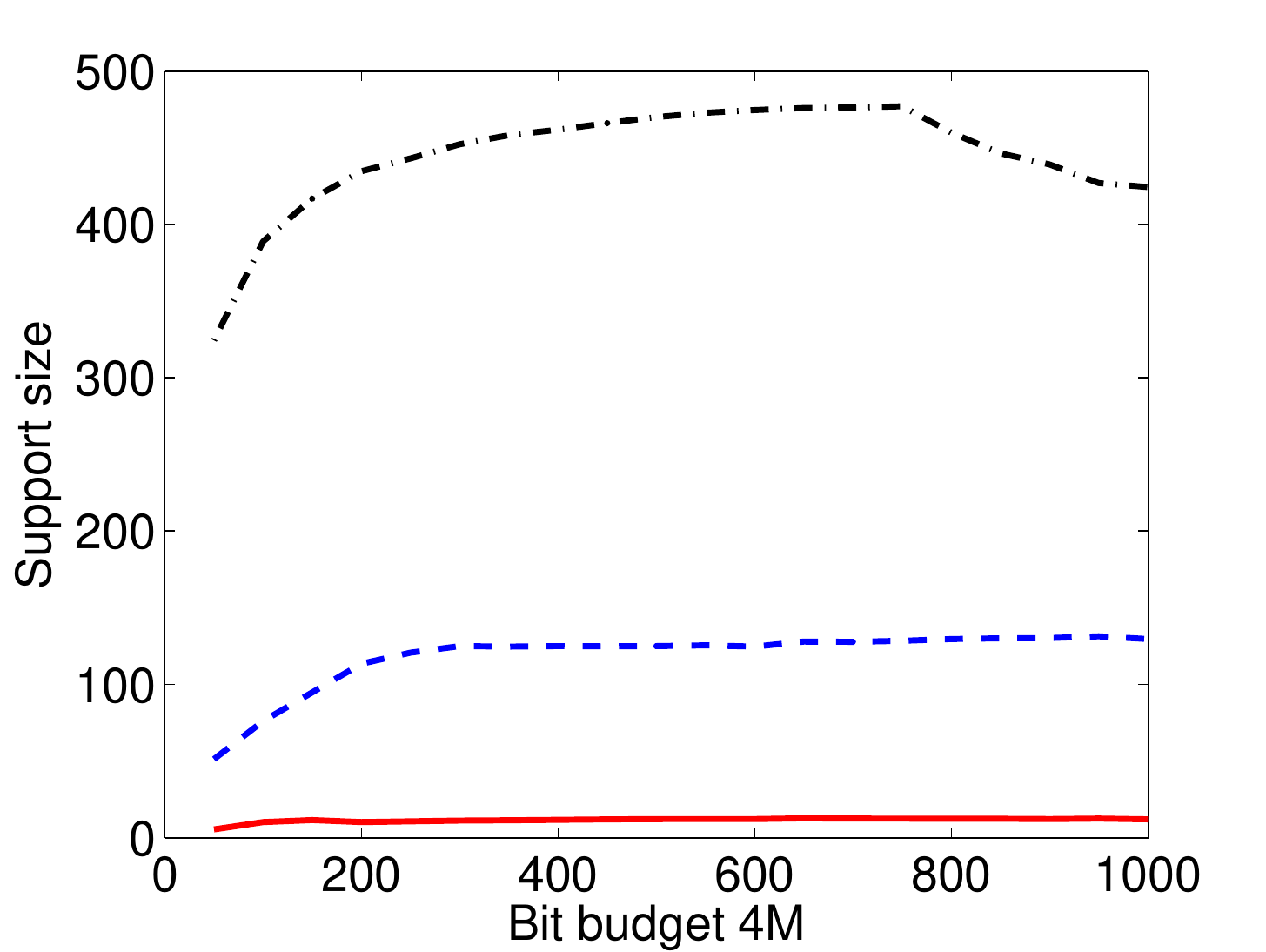}}%
  \subfigure[]{
    \label{Fig:Time_multibit_SNR30}
    \includegraphics[width=2.33in]{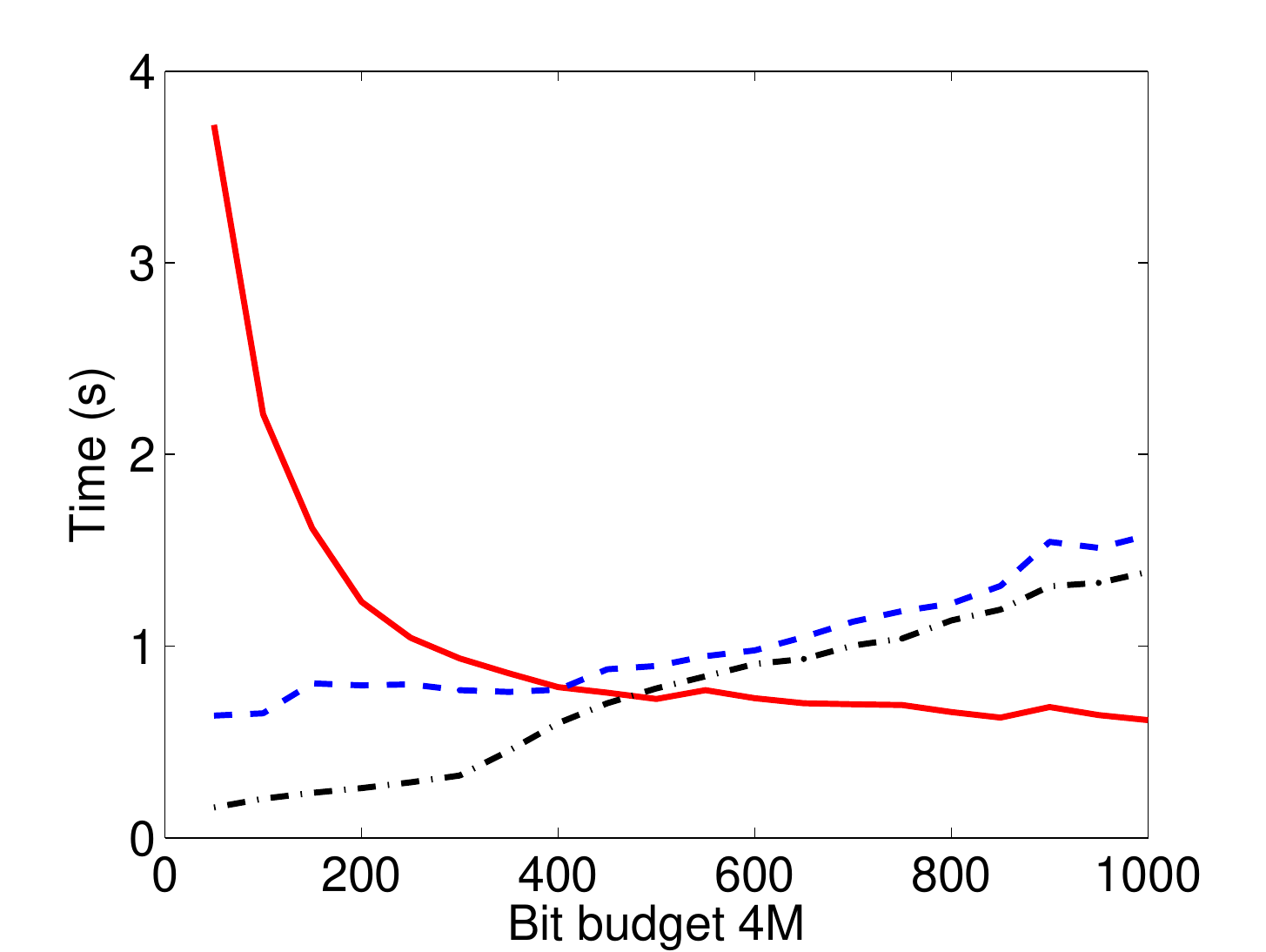}}
\centering
\caption{Performance comparison of Q-VMP, BPDN and L1RML with bit depth $B=4$. $\text{SNR}=30$dB. (a) Averaged reconstruction SNR; (b) Averaged support size of recovered signal; (c) Averaged CPU time.} \label{Fig:comparison_multibit}
\end{figure*}

In multi-bit CS, we first consider the case of a uniform unsaturated quantizer. As in the last subsection, we set $\text{SNR}=30$dB and $B=4$. Besides Q-VMP, we also use BPDN \cite{candes2006stable} and L1RML \cite{zymnis2010compressed} to recover the signal for comparison. Q-VMP is implemented as in the last subsection. BPDN solves the problem in (\ref{formu:BPDN}) with $\m{y}$ replaced by $\m{z}$
and is implemented using $\ell_1$-magic (available at http://users.ece.gatech.edu/$\sim$justin/l1magic). We set $\epsilon=\twon{\m{z}-\m{A}\m{x}}$ for achieving the best result though this value is unavailable in practice. In L1RML the regularization parameter is tuned such that it produces an estimate with the optimal RSNR. Additionally, we set $\tau=\frac{\sigma^2}{\twon{\m{A}}^2}$, $\epsilon=10^{-4}$ and $\beta=0.5$. Readers are referred to \cite{zymnis2010compressed} for their interpretations.

The experimental results are shown in Fig. \ref{Fig:comparison_multibit}, where red solid lines denote Q-VMP, black dashed dot lines denote BPDN, and blue dashed lines denote L1RML. Fig. \ref{Fig:RSNR_multibit_SNR30} depicts the averaged reconstruction SNRs of the three algorithms. A significant improvement of the reconstruction SNR can be observed using the proposed Q-VMP. It is over 6dB in comparison with L1RML and about an amplitude for BPDN. Moreover, Fig. \ref{Fig:SuppSize_multibit_SNR30} shows that Q-VMP produces the sparsest solution. Note that L1RML can produce a sparser solution by setting a larger regularization parameter \cite{zymnis2010compressed} but at the cost of a lower RSNR. Fig. \ref{Fig:Time_multibit_SNR30} shows that the speed of Q-VMP is comparable with that of BPDN and L1RML. Implemented with the basis pruning approach, Q-VMP is faster when more measurements are acquired since it is observed in such a case that the basis pruning approach works more efficiently.

\subsubsection{Saturated Quantizer}
We next consider the case of a saturated quantizer. We adopt the same experimental setup but a saturated quantizer where a noisy measurement falls in each quantization interval with the same probability. Since both the sensing matrix and measurement noise are Gaussian in the experiment, the noisy measurements are i.i.d. Gaussian $\cN\sbra{0,M^{-1}+\sigma^2}$. Then it is easy to get the quantizer. As a result, $12.5\%$ of the measurements are saturated in expectation. BPDN is inappropriate in such a case. We compare Q-VMP only with L1RML. The averaged reconstruction SNRs of Q-VMP and L1RML are presented in Fig. \ref{Fig:RSNR_multibit_saturationORnot} (red solid lines). Q-VMP obtains a RSNR of about 10dB higher than L1RML when sufficient measurements are acquired. The performance of the two algorithms on support size and speed is similar to that in the uniform quantizer case and is omitted.

\begin{figure}
\centering
\includegraphics[width=3.5in]{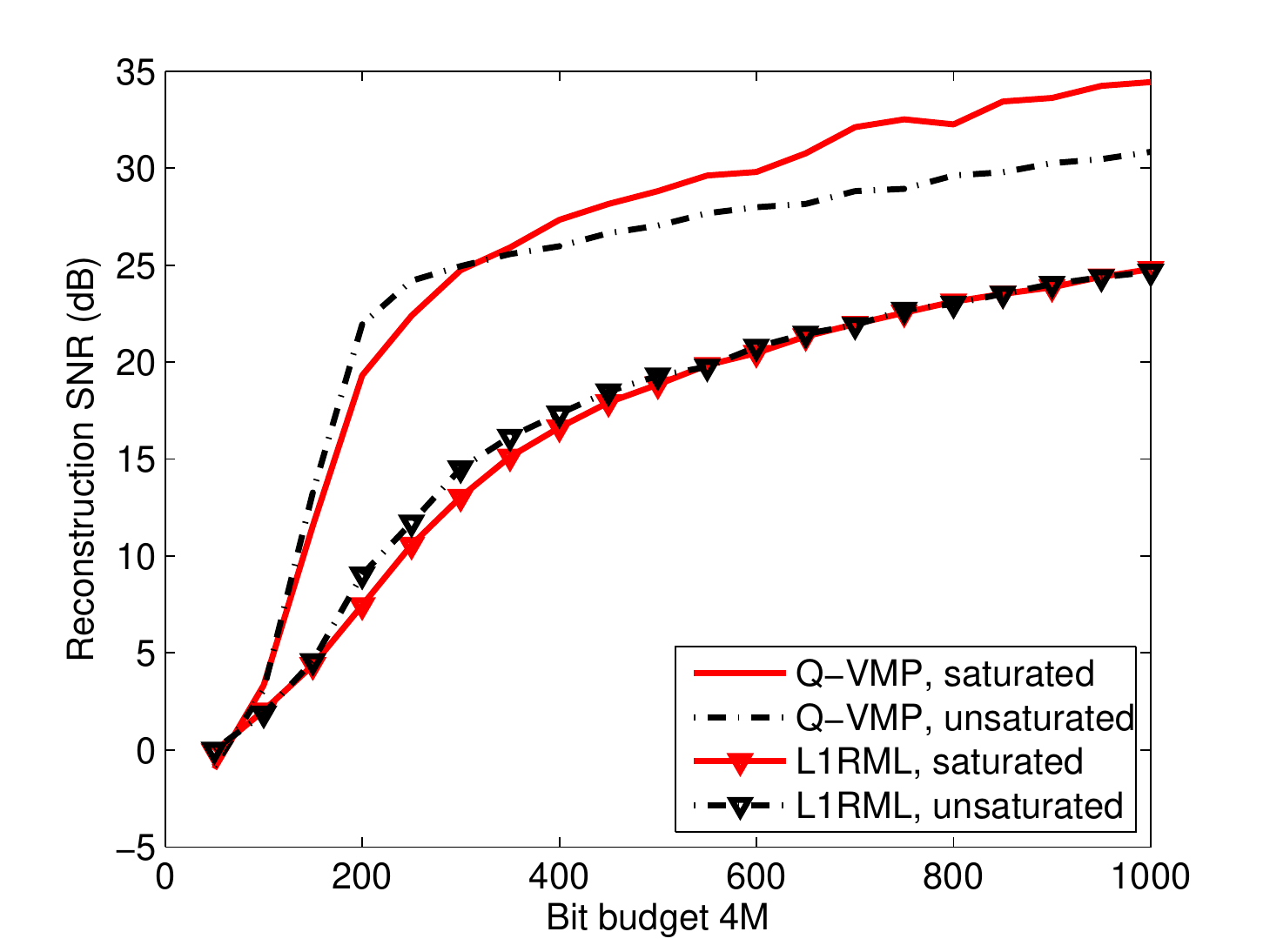}
\centering
\caption{Reconstruction SNRs of Q-VMP and L1RML with a saturated quantizer, as well as those with the unsaturated quantizer in Fig. \ref{Fig:comparison_multibit}.}
\label{Fig:RSNR_multibit_saturationORnot}
\end{figure}

The experiment above may shed light on the optimal quantizer design for Q-VMP. By comparing the performance of Q-VMP in the two quantizer scenarios, it can be seen from Fig. \ref{Fig:RSNR_multibit_saturationORnot} that the saturated quantizer outperforms the uniform unsaturated one when more measurements are taken for Q-VMP while it is not so clear for L1RML. We pose the problem of the optimal quantizer design for Q-VMP as a future work.

\subsection{Performance Comparison in 1-bit CS}
\begin{figure*}
\centering
  \subfigure[]{
    \label{Fig:RSNR_SNR10}
    \includegraphics[width=2.33in]{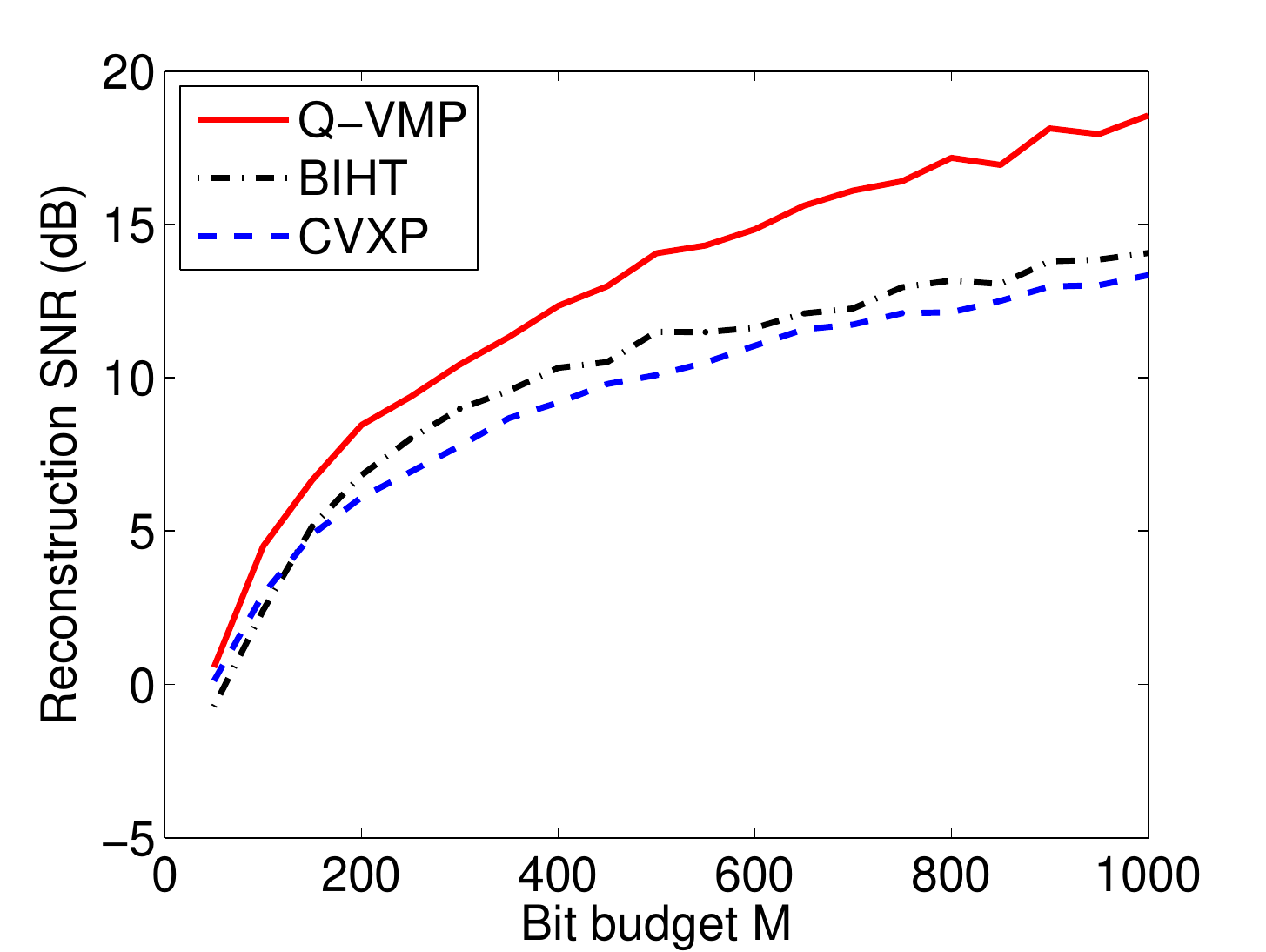}}%
  \subfigure[]{
    \label{Fig:SuppSize_SNR10}
    \includegraphics[width=2.33in]{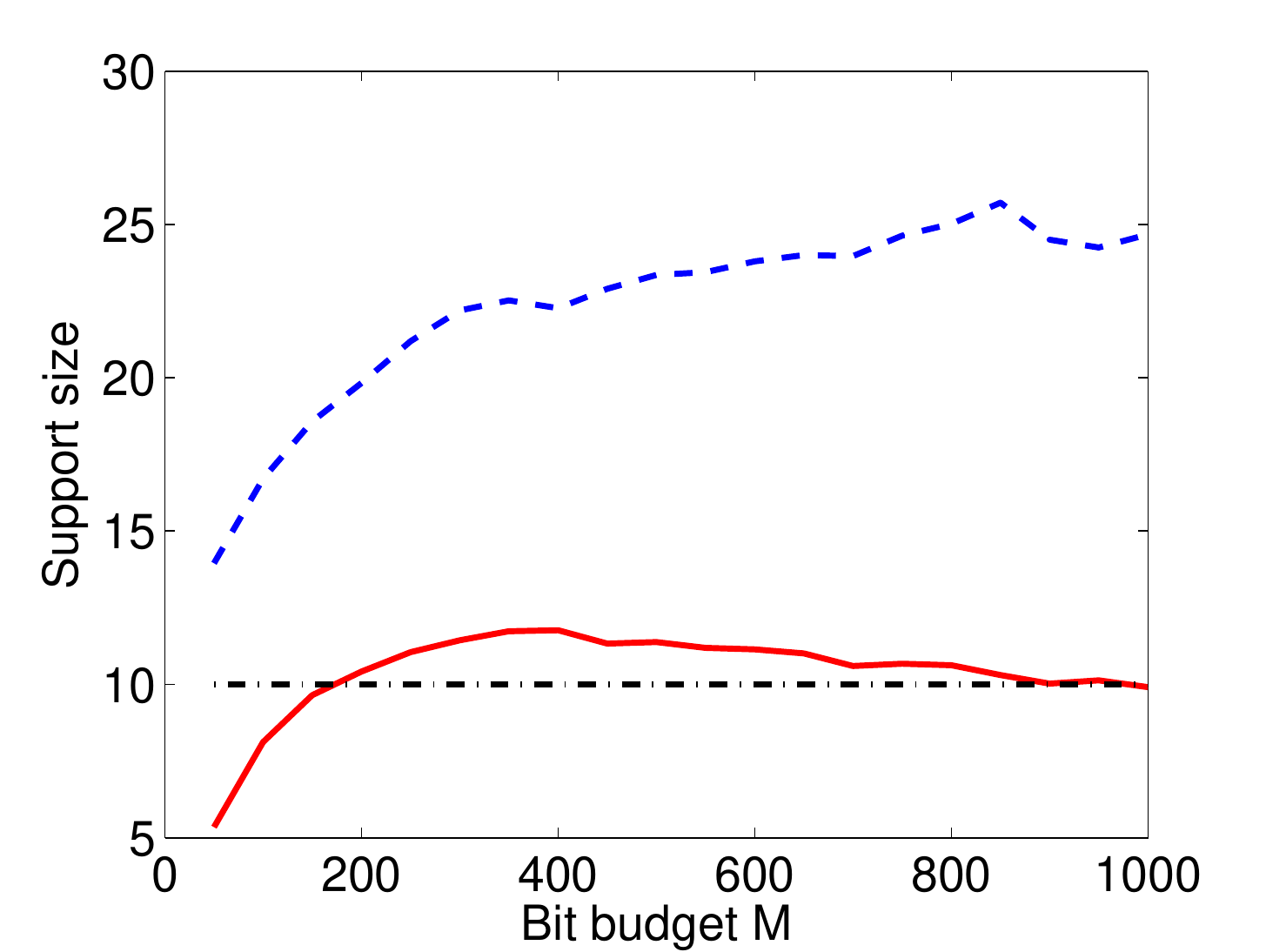}}%
  \subfigure[]{
    \label{Fig:Time_SNR10}
    \includegraphics[width=2.33in]{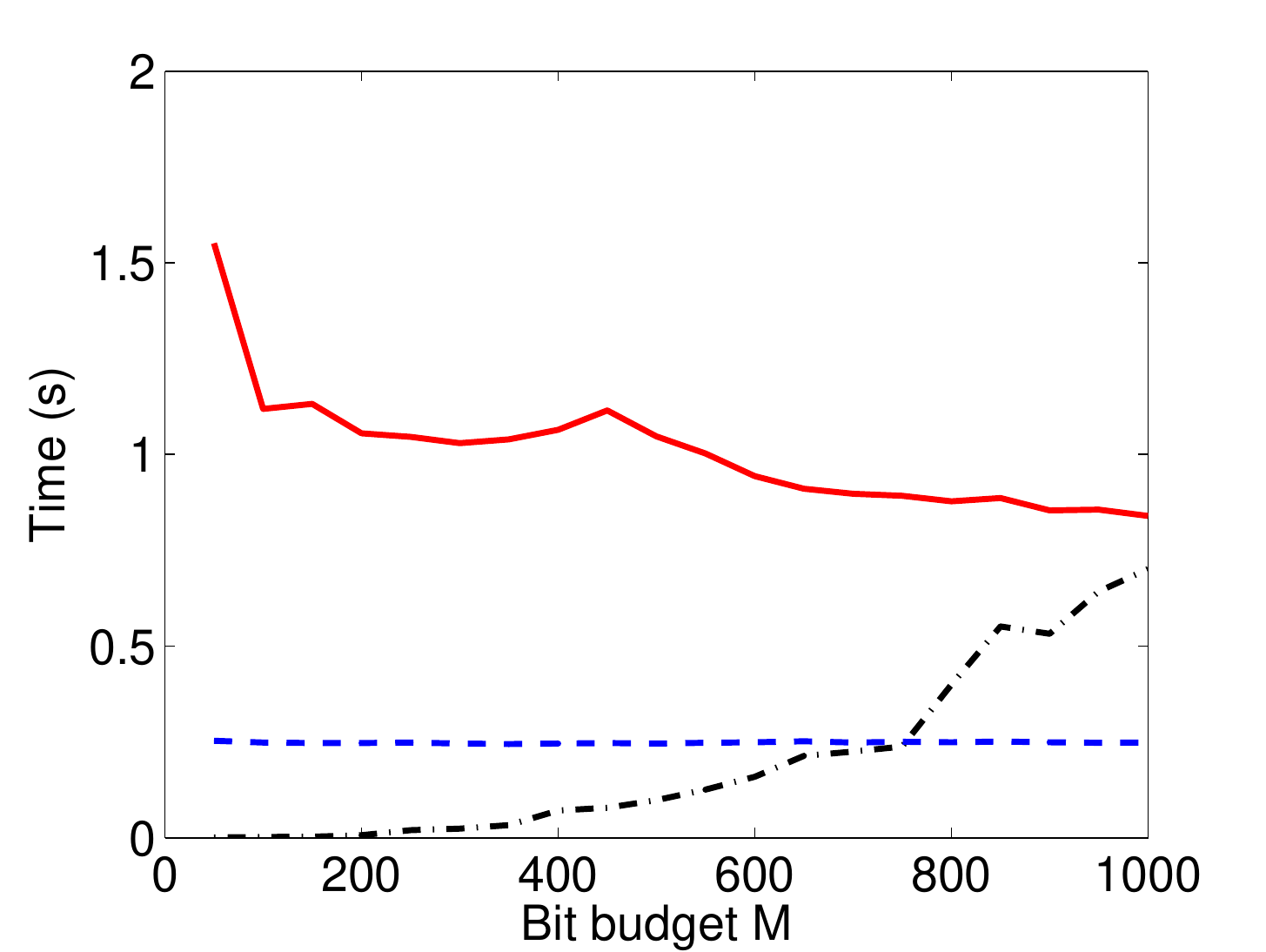}}
\centering
\caption{Performance comparison of Q-VMP, BIHT and CVXP in 1-bit CS. $\text{SNR}=10$dB. (a) Averaged reconstruction SNR; (b) Averaged support size of recovered signal; (c) Averaged CPU time.} \label{Fig:comparison_1bit}
\end{figure*}

The bit-depth $B=1$ in 1-bit CS. We set $\text{SNR}=10$dB. In such a case, $9.75\%$ measurements flip their signs due to the noise in expectation. We compare Q-VMP with the state-of-the-art algorithms BIHT \cite{jacques2011robust} and the convex programming approach in \cite{plan2012robust}, denoted by CVXP. The three algorithms are implemented as follows. In Q-VMP, We initialize $\inp{\alpha_n^{-1}}_{q\sbra{\m{\alpha}}}=\sqrt{M}/\abs{\m{A}_n^T\sgn\sbra{\m{z}}}$, $n=1,\cdots,N$, $\inp{\eta}_{q\sbra{\eta}}=1$ and $\inp{\m{e}}_{q\sbra{\m{e}}}=-\sgn\sbra{\m{z}}/\sqrt{M}$. As addressed in Remark \ref{rem:robust_1bit}, the effective noise level in 1-bit CS is much lower than the true one. We empirically find that it is a good choice to set the noise variance in Q-VMP to $10^{-3}\sigma^2$. We set $\tau_{pruning}=10^{4}$ and terminate Q-VMP as in multi-bit CS. The recovered signal is finally scaled to unit norm for comparison with the original one. For BIHT the oracle information of $K$ is used, i.e., BIHT is certain to return a reconstruction with $K$ nonzero entries. It is terminated if the Hamming error (see \cite{jacques2011robust}) of the current recovery is below the expected Hamming error or the maximum number of iterations, set to 1000, is reached. For CVXP the oracle information of $\onen{\m{x}}$ is used for achieving the best result and CVX \cite{grant2008cvx} is used for its implementation.

Our experimental results are presented in Fig. \ref{Fig:comparison_1bit}, where red solid lines denote Q-VMP, black dashed dot lines denote BIHT, and blue dashed lines denote CVXP. It is shown in Fig. \ref{Fig:RSNR_SNR10} that the proposed Q-VMP outperforms consistently the other two algorithms in the recovery accuracy. From Fig. \ref{Fig:SuppSize_SNR10}, it can be seen that Q-VMP produces a sparser solution than CVXP while BIHT uses this oracle information. Fig. \ref{Fig:Time_SNR10} shows that the computational speed is a disadvantage of Q-VMP.

\section{Conclusion and Future Work} \label{sec:conclusion}
The problem of sparse signal recovery from noisy quantized compressive measurements was studied in this paper. A Bayesian framework was presented that unifies the multi- and 1-bit CS problems and is applicable to the noisy environment and/or saturated quantizer. An algorithm was proposed based on variational Bayesian inference under the proposed framework. Numerical simulations were provided to demonstrate its improved signal recovery accuracy over the existing results.

A convex formulation of the noisy 1-bit CS problem has been studied in \cite{plan2012robust} with guaranteed signal recovery performance. This paper has introduced a different convex formulation (problem (\ref{formu:opt_1-bit}) with $f\sbra{\m{x}}=\onen{\m{x}}$ and $s=1$) that explicitly exploits the noise information and does not need the knowledge of the signal sparsity. One future work is to explore its theoretical guarantee. One drawback of Q-VMP is its high computational complexity due to an inversion of a high dimensional matrix at each iteration though it has been greatly alleviated with a basis pruning approach adopted in this paper. Thus another future work is to develop fast alternatives to the current implementation. Since the signal recovery accuracy in multi-bit CS is very different when a different quantizer is adopted, as shown in the present paper and in \cite{laska2011democracy}, to design the optimal quantizer that minimizes the signal recovery error is another interesting future research topic.

\section*{Appendix: Proof of Lemma \ref{lem:calculateprojection}}
It is easy to show the following equivalences:
\equ{\begin{split}&\m{e}^* = \cP_{\cD_{e}}\sbra{\m{v}}\\
&\Leftrightarrow\m{e}^* = \arg\min_{\m{e}\in\cD_e}\twon{\m{e}-\m{v}}\\
&\Leftrightarrow\m{e}^* = \arg\max_{\m{e}\in\cD_e} \m{v}^T\m{e}\\
&\Leftrightarrow -\sgn\sbra{\m{z}}\odot\m{e}^* = \arg\max_{\m{w}} f\sbra{\m{w}}=\overline{\m{v}}^T\m{w},\\
&\qquad\st \twon{\m{w}}=1 \text{ and }\m{w}\succeq\m{0}.
\end{split}}

1) $\cI$ is nonempty. Note that $\twon{\m{w}_{\cI}}\leq1$ and $\overline{\m{v}}_{\cI^c}\preceq\m{0}$. By the Cauchy inequality,
\equ{\begin{split}f\sbra{\m{w}}
&=\overline{\m{v}}_{\cI}^T\m{w}_{\cI} + \overline{\m{v}}_{\cI^c}^T\m{w}_{\cI^c} \\ &\leq\twon{\m{w}_{\cI}}\twon{\overline{\m{v}}_{\cI}} + \overline{\m{v}}_{\cI^c}^T\m{w}_{\cI^c}\\
&\leq\twon{\overline{\m{v}}_{\cI}}.\end{split}}
It is readily verified that the equality holds if $\m{w}$ is in the form of $-\sgn\sbra{\m{z}}\odot\m{e}^*$.

2) $\cI$ is empty, i.e., $\overline{\m{v}}\preceq\m{0}$. We prove the following result: $f\sbra{\m{w}}\leq\twon{\m{w}}\max\sbra{\overline{\m{v}}}=\max\sbra{\overline{\m{v}}}$. It is obvious that the equality holds if $\m{w}$ is in the form of $-\sgn\sbra{\m{z}}\odot\m{e}^*$.

The case of $M=1$ is trivial. We next prove the case of $M=2$ and then use induction to complete the proof. When $M=2$, substitute $w_1=\sqrt{1-w_2^2}$ into $f\sbra{\m{w}}$ and then
\equ{g\sbra{w_2}:=f\sbra{\sqrt{1-w_2^2},w_2}=\sqrt{1-w_2^2}\overline{v}_1+w_2\overline{v}_2.}
It is easy to show that $g'\sbra{w_2}\leq0$ if $0\leq w_2\leq\frac{\abs{\overline{v}_2}}{\twon{\overline{\m{v}}}}$, and $g'\sbra{w_2}\geq0$ if $\frac{\abs{\overline{v}_2}}{\twon{\overline{\m{v}}}}\leq w_2<1$. So the maximum of $g\sbra{w_2}$ can only be obtained at the boundary of the interval $\mbra{0,1}$, i.e., $f\sbra{\m{w}}\leq\max\sbra{g\sbra{0}, g\sbra{1}}=\max\sbra{\overline{\m{v}}}$.

Suppose the lemma holds when $M=n-1$ with $n>3$. We next show that it holds when $M=n$. Denote $\m{w}_{-1}=\mbra{w_2,\cdots,w_n}^T$, $\overline{\m{v}}_{-1}=\mbra{\overline{v}_2,\cdots,\overline{v}_n}^T$. By $w_1^2+\twon{\m{w}_{-1}}^2=\twon{\m{w}}^2$ and applying the results when $M=n-1$ and $M=2$ consecutively,
\equ{\begin{split}f\sbra{\m{w}}
&=\overline{v}_1w_1+\overline{\m{v}}_{-1}^T\m{w}_{-1}\\
&\leq \overline{v}_1w_1+\max\sbra{\overline{\m{v}}_{-1}}\twon{\m{w}_{-1}}\\
&\leq \twon{\m{w}}\max\sbra{\overline{\m{v}}}=\max\sbra{\overline{\m{v}}}.\end{split}}

\bibliographystyle{IEEEtran}

\end{document}